\documentclass[conference,10pt]{IEEEtran}
\IEEEoverridecommandlockouts
\usepackage{cite,comment}
\usepackage{amsmath,amssymb,amsfonts}
\usepackage{algorithmic}
\usepackage{graphicx}
\usepackage{textcomp}
\usepackage{xcolor}
\usepackage{amsthm}

\theoremstyle{plain}
\newtheorem{theorem}{Theorem}
\newtheorem{corollary}{Corollary}
\newtheorem{lemma}{Lemma}
\newtheorem{proposition}{Proposition}

\theoremstyle{definition}
\newtheorem{definition}{Definition}

\theoremstyle{remark}

\newcommand{\Cech}{\v{C}ech}

\def\BibTeX{{\rm B\kern-.05em{\sc i\kern-.025em b}\kern-.08em
    T\kern-.1667em\lower.7ex\hbox{E}\kern-.125emX}}

\begin{document}

\title{Fundamental Limits of Quantum Semantic Communication via Sheaf Cohomology}
\author{\IEEEauthorblockN{Christo Kurisummoottil Thomas}
\IEEEauthorblockA{\textit{Electrical and Computer Engineering} \\
\textit{Worcester Polytechnic Institute,Worcester, MA, USA}
 \\
cthomas2@wpi.edu}
\and
\IEEEauthorblockN{Mingzhe Chen}
\IEEEauthorblockA{\textit{Electrical and Computer Engineering} \\
\textit{University of Miami,
Miami, FL, USA} \\
mingzhe.chen@miami.edu}
}
\maketitle
\vspace{-18mm}
\begin{abstract}
Semantic communication (SC) enables bandwidth-efficient coordination in multi-agent systems by transmitting meaning rather than raw bits. However, when agents maintain heterogeneous internal representations shaped by distinct sensing modalities and artificial intelligence (AI) architectures, perfect bit-level transmission no longer guarantees mutual understanding. Despite significant progress in deep learning approaches for semantic compression, the information-theoretic foundations characterizing fundamental limits of semantic alignment  under heterogeneous agents remain incomplete.  Notably, quantum contextuality shares the same mathematical structure as semantic ambiguity, as both correspond to cohomological obstructions, motivating a quantum framework.
Inspired by this, in this paper, an information-theoretic framework for quantum SC is developed using sheaf cohomology. Multi-agent semantic networks are modeled as quantum sheaves, where agents' meaning spaces are Hilbert spaces connected by quantum channels. The first sheaf cohomology group $H^1$ is shown to precisely characterize irreducible semantic ambiguity, i.e., the fundamental obstruction to perfect alignment that no local processing can resolve. The minimum communication rate for semantic alignment is proven to equal $\log_2 \dim H^1$, establishing a semantic analog of Shannon's fundamental limits. For entanglement-assisted channels, the capacity is shown to strictly exceed classical bounds, with each ebit of shared entanglement saving one bit of classical communication. This provides a rigorous mechanism for the ``shared context'' often assumed in classical SC. Additionally, quantum contextuality is established as a resource that reduces cohomological obstructions, and a duality between quantum discord and integrated semantic information is proven, linking quantum correlations to irreducible semantic content. The framework provides rigorous foundations for quantum-enhanced SC in autonomous multi-agent systems.

\end{abstract}


\vspace{-2mm}\section{Introduction}
Semantic communication (SC) has emerged as a paradigm shift in communication theory, moving beyond Shannon's classical framework of transmitting bits to transmitting meaning \cite{KountourisCommMag2021}. This approach promises bandwidth savings in multi-agent systems where users share common objectives but may maintain heterogeneous internal representations shaped by distinct sensing modalities and artificial intelligence (AI) architectures. However, this heterogeneity poses a fundamental challenge: when users interpret symbols through different semantic lenses, perfect bit-level transmission no longer guarantees mutual understanding.

Consider two robots coordinating a manipulation task. Each robot maintains an internal world model encoding spatial relationships, object properties, and action affordances. When one robot communicates ``move left,'' the meaning depends on reference frames, learned associations, and task context that may differ between agents. The robots face an irreducible ambiguity stemming not from channel noise, but from their heterogeneous representations. This motivates our central questions: \emph{What is the minimum communication rate to resolve such ambiguity? Can shared quantum resources help?}
Classical SC models meanings as vectors in real spaces, with alignment achieved when agents' representations, projected onto a shared semantic space, agree \cite{pandolfo2025latent}. But this framework cannot capture contextuality, where a symbol's meaning depends on what other symbols are jointly queried and the history of prior exchanges. Static vector projections \cite{pandolfo2025latent} cannot represent such dependencies. Mathematically, contextuality manifests as a \emph{topological obstruction}: local consistencies that cannot be extended globally. Abramsky and Brandenburger \cite{AbramskyBrandenburger2011} formalized this insight, showing that \emph{quantum contextuality} corresponds precisely to obstructions in sheaf cohomology.
 By extending semantic spaces to complex Hilbert spaces and semantic channels to quantum operations \cite{Nielsen2010}, we gain access to this rich structure. The resulting quantum framework provides both a resource for alignment and a diagnostic tool: entanglement reduces required communication, while sheaf cohomology characterizes irreducible obstructions.

\vspace{-2mm}\subsection{Related Work}\vspace{-1mm}
Despite progress in deep learning for semantic compression \cite{bourtsoulatze2019deep,XieTSP2021,Gunduz2022,wu2024deep}, the information-theoretic foundations of SC remain incomplete. Semantic information theory, from Carnap and Bar-Hillel \cite{Carnap52} to Floridi \cite{Floridi2011PhilosophyOfInformation}, addresses what semantic content is. A recent work \cite{Lu2025} generalizes Shannon's rate-distortion theory by replacing distortion constraints with semantic constraints via truth functions. But these works \cite{Floridi2011PhilosophyOfInformation,Lu2025} do not address how to align heterogeneous semantic interpretations across users. The missing piece is a mathematics for \emph{semantic consistency} that characterizes when locally coherent agents can achieve global alignment, and quantifying obstructions when they cannot. \emph{Sheaf theory}, recently applied to multi-agent networks \cite{Hansen2019SpectralSheaves,grimaldi2025learning,ghalkha2025sheafalign}, formalizes semantic consistency through global sections representing alignment and cohomology measuring irreducible ambiguity.  The sheaf-theoretic approach to quantum contextuality \cite{AbramskyBrandenburger2011} was further developed by \cite{Abramsky2012CohomologyNonlocality} and \cite{Caru2017CohomologyContextuality} using \emph{sheaf cohomology}, an algebraic topological tool that rigorously measures obstructions to extending local consistency to global agreement.
 Entanglement-assisted classical capacity was established by \cite{Bennett2002EntanglementAssisted}; quantum SC was proposed by \cite{Chehimi2024QuantumSemantic}. Unlike prior work that assumes shared semantics or proposes SC protocols without capacity theorems, we establish fundamental limits.

\vspace{-2mm}\subsection{Our  Contributions}\vspace{-1mm}
The main contribution of this paper is a rigorous framework proving that sheaf cohomology precisely characterizes irreducible semantic ambiguity, and that shared entanglement reduces the communication rate required to align heterogeneous agents. Our key contributions include:
\begin{itemize}
    \item \emph{Cohomological semantic rate (Theorem 1):} We prove that the minimum communication rate for perfect semantic alignment equals $\dim H^1$, the dimension of the first sheaf cohomology group.  Intuitively, $\dim H^1$
counts the number of independent semantic ambiguities that cannot be resolved by local processing alone.  This provides a semantic analog of Shannon's channel coding theorem \cite{Shannon1948}.
    
    \item \emph{Entanglement-assisted capacity (Theorem 2):} We derive semantic channel capacity when sender and receiver share prior entanglement, proving it strictly exceeds classical capacity. Entanglement provides a physical realization of the ``shared context" often assumed in classical SC, with each shared entangled pair saving one bit of alignment communication. The quantum advantage grows as channel noise increases, which is precisely when SC most benefits over bit-level transmission \cite{ChaccourArxiv2022}.
    
    \item \emph{Contextuality as resource (Theorem 3):} We show that quantum contextuality reduces cohomological obstructions to semantic alignment. Contextual correlations act as ``pre-shared semantic resolution,'' establishing contextuality as a resource for SC.
    
    \item \emph{Discord-integration duality (Theorem 4):} We prove that quantum discord, which measures quantum correlations beyond entanglement, equals integrated semantic information. This links quantum correlations to irreducible semantic content and connects our framework to integrated information theory \cite{Tononi2016}.
\end{itemize}

\vspace{-1mm}\section{Modeling Semantic Networks as Quantum Sheaves}
\vspace{-1mm}
This section establishes a mathematical framework for quantum SC by extending network sheaves from classical vector spaces to quantum Hilbert spaces, thereby enabling the use of tools from algebraic topology and quantum information theory \cite{Nielsen2010}.

\vspace{-2mm}\subsection{Semantic Networks and Classical Sheaves}\vspace{-1mm}
Consider a network of communicating agents represented by a graph $G = (V, E)$, where vertices $V$ correspond to agents and edges $E$ represent communication links. Each agent $v\in V$ maintains an internal semantic representation, which is a structured encoding of meanings, concepts, and their relationships \cite{ChristoTWCArxiv2022}. When agents communicate, they must align these heterogeneous representations to achieve mutual understanding.
A \emph{network sheaf} over $G$ assigns a vector space to each vertex and a linear map to each edge, capturing how semantic information transforms as it passes between agents. Formally, a classical network sheaf consists of:
\vspace{-1mm}\begin{itemize}
    \item A \emph{stalk space} $\mathcal{F}_v$ at each vertex $v$, representing the agent's semantic state space, encompassing the set of semantic concepts.
    \item A \emph{restriction map} $F_e: \mathcal{F}_u \to \mathcal{F}_v$ for each edge $e = (u, v)$, representing how agent $u$'s semantics project onto the communication channel to agent $v$.
\vspace{-1mm}\end{itemize} The transformation occurs because agents may use different coordinate systems, communication may compress or project information, or semantic ``translation'' between representations is required. A \emph{local section} over a subset $U \subseteq V$ is an assignment of semantic states that is consistent along edges: for every edge between vertices in $U$, the restriction maps agree. A \emph{global section} extends this consistency to the entire network, representing perfect semantic alignment among all agents.
The key insight from sheaf theory is that global sections may fail to exist even when local sections are everywhere consistent, a topological phenomenon we formalize in Section~\ref{Obstruction}. However, classical sheaves have fundamental limitations: they cannot represent agents holding probabilistic mixtures over meanings, nor capture correlations shared between agents prior to communication. These considerations motivate extending the framework to quantum Hilbert spaces.

\vspace{-1mm}\subsection{Quantum Semantic Sheaves}\vspace{-1mm}
We extend the classical framework to the quantum setting, motivated by two observations. First, semantic representations in modern AI systems such as neural embeddings, attention weights, latent spaces are naturally described by vectors in high-dimensional spaces, and quantum mechanics provides a canonical framework for such spaces. Second, quantum systems exhibit contextuality, which Abramsky and Brandenburger showed corresponds precisely to sheaf-theoretic obstruction~\cite{AbramskyBrandenburger2011} .

\vspace{-2mm}\begin{definition}
A \emph{quantum semantic sheaf} over a communication graph
$G = (V,E)$ is a tuple $(\mathcal{H}, \mathcal{F}, \rho)$ where:
\begin{enumerate}[(i)]
    \item $\mathcal{H}_v$ is a finite-dimensional Hilbert space at each vertex
    $v$, representing agent $v$'s semantic state space.
    
    \item $\mathcal{F}_e : \mathcal{L}(\mathcal{H}_u) \to \mathcal{L}(\mathcal{H}_v)$
    is a completely positive trace-preserving (CPTP) map for each edge
    $e \!=\! (u,v)$, representing the semantic channel from agent $u$ to  $v$.
    
    \item $\rho_v$ is a density operator on $\mathcal{H}_v$ representing
    agent $v$'s current semantic state. A density operator is a positive semi-definite matrix with $\mathrm{Tr}(\rho) = 1$
, generalizing probability distributions to quantum systems \cite{Nielsen2010}.
\vspace{-1mm}\end{enumerate}
\end{definition}

The CPTP maps $F_e$ serve a dual role: they model both the representational transformation between agent $u$'s local meaning space and the shared comparison space, and the physical quantum channel through which semantic states are communicated. Channel noise, loss, and decoherence are encoded directly in the sheaf structure. The density operator $\rho_v$ encodes agent $v$'s probabilistic belief distribution over semantic concepts, with diagonal elements representing probabilities of discrete meanings and off-diagonal elements capturing quantum coherences between superposed interpretations. This generalizes classical network sheaves: when states are pure ($\rho_v = |\psi\rangle\langle\psi|$) and channels are unitary ($F_e(\rho) = U_e \rho U_e^\dagger$), we recover the classical setting. The quantum generalization admits two phenomena absent classically: shared entanglement as a resource, and contextuality-induced obstructions. A \emph{semantic configuration} is an assignment $\sigma = \{\sigma_v\}_{v \in V}$ of density operators $\sigma_v \in \mathcal{D}(\mathcal{H}_v)$ to each vertex. Two agents $u$ and $v$ connected by edge $e$ are \emph{semantically aligned} if:$
F_e(\sigma_u) = \sigma_v.$
\vspace{-2mm}\begin{definition}
A \emph{local section} over a subset of agents $U \subseteq V$ is a semantic configuration
$\sigma$ such that $
\mathcal{F}_e(\sigma_u) = \sigma_v$, 
for all edges $e = (u,v)$ with both endpoints in $U$.
A \emph{global section} is a local section over the entire vertex set $V$.
\vspace{-1mm}\end{definition}
When a global section exists, the network achieves perfect semantic
alignment. This means that all agents share a consistent semantic state that transforms
coherently along every communication link. The central question we address here is: \emph{when does
such alignment exist, and what communication resources are required to achieve it?} To answer this, we first need tools to measure how far a configuration is from perfect alignment, and to characterize obstructions when perfect alignment is impossible.

\vspace{-1mm}\subsection{The Sheaf Laplacian}\vspace{-1mm}

To quantify deviation from semantic alignment, we employ the \emph{sheaf Laplacian}~\cite{Hansen2019SpectralSheaves}, which generalizes the graph Laplacian and measures how far a semantic configuration deviates from being a global section.
First, we define
the coboundary operator $\delta^0$ that maps vertex assignments to edge disagreements. For a semantic configuration $\sigma$, define
\begin{equation}
    (\delta^0 \sigma)_e = \mathcal{F}_e(\sigma_u) - \sigma_v,
    \label{eq_coboundary}
\end{equation}
for each edge $e = (u,v)$. The configuration $\sigma$ is a global section if and only if $\delta^0 \sigma = 0$.
\vspace{-2mm}\begin{definition}
The \emph{sheaf Laplacian} is defined as
\begin{equation}
\mathcal{L} = (\delta^0)^\dagger \delta^0,
\vspace{-1mm}\end{equation}
where $\dagger$ denotes the {adjoint operator} with respect to the Hilbert--Schmidt inner product $\langle A, B \rangle_{\text{HS}} = \text{Tr}(A^\dagger B)$. The adjoint $(\delta^0)^\dagger$ is the unique linear operator satisfying the duality condition:
$
\langle (\delta^0)^\dagger x, y \rangle_{\text{HS}} = \langle x, \delta^0 y \rangle_{\text{HS}},
$
for all semantic configurations $x, y$. The operator $\delta^0$ maps vertex assignments to edge disagreements; $(\delta^0)^\dagger$ reverses this flow by aggregating edge errors back to vertices. The kernel of $\mathcal{L}$ consists precisely of the global sections:
\vspace{-1mm}\begin{equation}
\ker(\mathcal{L}) = \{\sigma : F_e(\sigma_u) = \sigma_v \text{ for all edges } e = (u,v)\}.
\vspace{-1mm}\end{equation}
\end{definition}

The eigenvalues of $\mathcal{L}$ quantify the ``cost'' of semantic misalignment. The smallest nonzero eigenvalue, known as the \emph{spectral gap}, determines how quickly local alignment protocols converge to global consensus.
\vspace{-1mm}\subsection{ Cohomological Obstruction}\vspace{-1mm}
\label{Obstruction}

To formally characterize irreducible semantic obstructions, we construct the \emph{\Cech{} complex} \cite{Hatcher2002} adapted to our quantum setting:
\vspace{-1mm}\begin{equation}
      C^{0}(G,S) \xrightarrow{\ \delta^{0}\ } C^{1}(G,S) \xrightarrow{\ \delta^{1}\ } C^{2}(G,S) \xrightarrow{\ \delta^{2}\ } \cdots,
\vspace{-1mm}\end{equation}
where \(C^{0}\) consists of vertex assignments that are semantic states, \(C^{1}\) consists of edge
assignments that are semantic disagreements along communication links, and the coboundary maps \(\delta^{k}\) measure inconsistency at each hierarchical level (shared semantics among $k+1$ set of agents).   
 Cohomology groups are further defined as the quotient of cycles (elements in the kernel) by boundaries (elements in the image), capturing obstructions that cannot be removed by local adjustments. For semantic alignment:
\vspace{-2mm}\begin{align}
H^0(G, \mathcal{S}) &= \ker(\delta^0) & &\text{(global sections)}, \\
H^1(G, \mathcal{S}) &= \ker(\delta^1) / \mathrm{im}(\delta^0) & &\text{(obstruction space)}.
\vspace{-2mm}\end{align}
The zeroth cohomology $H^0$ captures what global alignments exist. The first cohomology $H^1$ captures irreducible obstructions, which are semantic ambiguities that cannot be resolved by any local processing. A nonzero element of $H^1$ represents a configuration that is locally consistent on every edge but cannot be extended to a global section. To illustrate this, consider three agents $A, B, C$ in a loop, each maintaining different semantic reference frames. A 1-cocycle assigns frame disagreements to edges: $(d_{AB}, d_{BC}, d_{CA})$. If these satisfy $d_{AB} + d_{BC} + d_{CA} = 0$ (loop consistency), they form a cycle. However, if no global choice of reference frames can simultaneously satisfy all edges, then this cycle represents a topological incompatibility, i.e., an element of $H^1 \neq 0$. This obstruction cannot be removed by local communication; it reflects fundamental misalignment in the network structure.
\begin{proposition}
A quantum semantic sheaf admits a global section for every locally consistent configuration if and only if $H^1(G, \mathcal{S}) = 0$.
\end{proposition}
\begin{IEEEproof}
    This follows from the standard characterization of sheaf cohomology \cite{Abramsky2012CohomologyNonlocality}, which extends to our quantum setting since CPTP maps preserve the algebraic structure required for the Čech complex construction.
\end{IEEEproof}

\subsection{ Entanglement-Assisted Semantic Sheaves}

We now incorporate entanglement as a resource for SC. 

\begin{definition}
An \emph{entanglement-assisted quantum semantic sheaf} augments Definition~1 with a bipartite entangled state $\Phi_e \in \mathcal{D}(\mathcal{H}'_u \otimes \mathcal{H}'_v)$ for each edge $e = (u, v)$, where $\mathcal{H}'_u$ and $\mathcal{H}'_v$ are local ancilla spaces held by agents $u$ and $v$ respectively. The state $\Phi_e$ is \emph{maximally entangled} if it achieves maximum entanglement entropy of $\log_2 d$ ebits for $d$-dimensional systems (e.g., Bell states for qubits).
\end{definition}
The semantic channel $F_e$ is now \emph{assisted}: agent $u$ can perform joint operations on $\mathcal{H}_u \otimes \mathcal{H}'_u$ before transmission, and agent $v$ can perform joint operations on the received state together with $\mathcal{H}'_v$.


\begin{proposition}
For an entanglement-assisted semantic sheaf with maximally entangled states on each edge:
\vspace{-1mm}\begin{equation}
  \dim H^{1}_{\mathrm{EA}}(G,S) \;\le\; \dim H^{1}(G,S),
\vspace{-1mm}\end{equation}
with strict inequality for generic sheaves over graphs with cycles.
\end{proposition}
\vspace{-2mm}\begin{IEEEproof}
    Maximally entangled states enable superdense coding, effectively doubling the classical capacity of each edge. This additional capacity can eliminate cocycles in $H^1$
that arise from holonomy obstructions around loops. 
\end{IEEEproof}

The intuition is that entanglement provides a ``pre-shared reference frame'' that resolves certain semantic ambiguities without explicit communication.  Our framework assumes pre-shared entanglement without accounting for its generation and distribution cost. Quantifying the net advantage after entanglement overhead remains an important open problem.

\vspace{-1mm}\subsection{ Connection to Contextuality}\vspace{-1mm}
Our framework connects directly to quantum contextuality~\cite{AbramskyBrandenburger2011}. In \cite{AbramskyBrandenburger2011}, a measurement scenario defines a sheaf over a simplicial complex of compatible measurements, and contextuality corresponds to the nonexistence of global sections.
\vspace{-5mm}\begin{proposition}
Let $\mathcal{M}$ be a quantum measurement scenario with sheaf $\mathcal{S}_\mathcal{M}$ whose vertices correspond to measurement contexts and whose stalks assign outcome probabilities. There exists a quantum semantic sheaf $\mathcal{S}$ such that $H^1(\mathcal{S}) \cong H^1(\mathcal{S}_\mathcal{M})$. Contextual measurement scenarios correspond precisely to semantic sheaves with irreducible ambiguity.
\vspace{-2mm}\end{proposition}
This correspondence reveals that quantum contextuality, a fundamental feature distinguishing quantum from classical physics manifests as semantic ambiguity in multi-agent communication. Contextual correlations represent semantic content that transcends classical description.

\vspace{-1mm}\section{Main Results: Semantic Rate and Entanglement-Assisted Capacity}
Having established the mathematical framework, we now prove our main results characterizing the fundamental limits of quantum SC.

\vspace{-2mm}\subsection{Cohomological Semantic Rate}\vspace{-1mm}

Our first main result establishes that cohomology determines the fundamental rate limit for semantic alignment.

\vspace{-2mm}\begin{theorem}
\label{theorem_rate}
For a quantum semantic sheaf $\mathcal{S}$ over graph $G$, the minimum communication rate required for perfect semantic alignment is:
\vspace{-1mm}\begin{equation}
R^*_{\mathrm{sem}} = \log_2 \dim H^1(G, \mathcal{S})
\vspace{-1mm}\end{equation}
measured in qubits (or bits for classical sheaves).
\end{theorem}
\begin{IEEEproof}
See Appendix A of extended version in arxiv \cite{ThomasISIT2026}.
\end{IEEEproof}
The rate $R^*_{\mathrm{sem}}$ represents the irreducible communication cost imposed by representational heterogeneity, which is defined as the minimum bits required to resolve semantic ambiguity regardless of protocol design.

\vspace{-2mm}\begin{corollary}\label{cor:capacity}
The semantic capacity of a classical channel with capacity $C$ is:
\vspace{-3mm}\begin{equation}
C_{\mathrm{sem}} = \frac{C}{\log_2\dim H^1}
\vspace{-1mm}\end{equation}
semantic messages per channel use (when $H^1 \neq 0$).
\end{corollary}
\vspace{-2mm}\begin{IEEEproof}
    Each semantic alignment requires $R^{\ast} = \log_2 \,\text{dim}\, H^1$ bits (Theorem 1). A channel with capacity $C$ transmits $C$ bits per use. Therefore, $C_{\mathrm{sem}} =\frac{C}{R^{\ast}} = \frac{C}{\log_2\dim H^1}$.
\end{IEEEproof}
These results quantify the fundamental cost of meaning alignment: $\log_2 \text{dim}\, H^1$ bits must be communicated to resolve semantic ambiguity, independent of protocol design. This is the irreducible overhead imposed by representational heterogeneity. Theorem~1 characterizes the structural rate requirement under idealized assumptions of noiseless channels and unlimited computational resources. Practical systems face additional constraints: (1) physical channel capacity $C$ bounding transmission rate via the Holevo limit \cite{Holevo1973}, and (2) cognitive/computational limits on agents' encoding and decoding complexity. When $C < \log_2 \dim H^1$, perfect semantic alignment becomes impossible regardless of protocol design.

\vspace{-2mm}\subsection{Entanglement-Assisted Semantic Capacity}\vspace{-2mm}

We now establish that entanglement provides provable advantage for SC.
\begin{theorem}
\label{EntanglementCapacity}
For a quantum semantic sheaf $\mathcal{S}$ over graph $G = (V, E)$ with entanglement-assisted sheaf $\mathcal{S}_{EA}$, where edge $e$ shares entangled state $\Phi_e$ with Schmidt rank $r_e$. The Schmidt rank is the number of terms in the Schmidt decomposition of a bipartite state, quantifying its entanglement. A product state has $r=1$; a maximally entangled $d$-dimensional state has $r=d$:
\vspace{-3mm}\begin{equation}
\dim H^1_{EA}(G, \mathcal{S}) = \max\left(0, \dim H^1(G, \mathcal{S}) - \sum_{e \in E} \log_2 r_e\right).
\vspace{-1mm}\end{equation}
For separable (classical) correlations where $r_e = 1$ for all $e$, no reduction occurs: $\dim H^1_{EA} = \dim H^1$.
\end{theorem}

\vspace{-1mm}\begin{IEEEproof}
    See Appendix B of extended version in arxiv \cite{ThomasISIT2026}.
\vspace{-1mm}\end{IEEEproof}
Theorem~\ref{EntanglementCapacity} establishes a precise tradeoff between entanglement and communication: each ebit of shared entanglement saves one bit of classical communication for semantic alignment. In this sense, entanglement provides a rigorous, physically realizable mechanism for the ``shared knowledge bases" or ``common context" often assumed in classical SC and is quantified exactly by Schmidt rank.


These results establish a precise connection between algebraic topology and information theory.  The key distinction is that Shannon theory concerns \emph{statistical} uncertainty (entropy), while our framework concerns \emph{structural} uncertainty (cohomology) which are obstructions arising from the topology of heterogeneous representations rather than probabilistic noise.

\vspace{-3mm}\section{Contextuality and Integration}\vspace{-1mm}
Building on the cohomological framework of Sections II and III, this section establishes two key connections: quantum contextuality serves as a resource for SC by reducing cohomological obstruction, and quantum discord precisely equals integrated semantic information. \emph{Integrated semantic information} captures holistic meaning that exists only in the joint system of agents, namely semantic content that cannot be decomposed into what each agent knows independently. For example, the relational meaning ``A is left of B” requires both agents and cannot be split into A’s part plus B’s part.

\vspace{-2mm}\subsection{Contextuality as Semantic Resource}\vspace{-1mm}

Quantum contextuality, which is defined as the impossibility of assigning pre-existing values to all observables independent of measurement context—is a defining feature of quantum mechanics. \cite{AbramskyBrandenburger2011} showed that contextuality corresponds exactly to the nonexistence of global sections in a presheaf over measurement contexts. We now show that contextual correlations reduce semantic obstruction.
\vspace{-2mm}\begin{definition}
A quantum semantic sheaf $\mathcal{S}$ exhibits \emph{contextual correlations} if there exists no classical, i.e., local hidden variable model that reproduces the joint distributions over all semantic measurements. Formally, $\mathcal{S}$ is contextual if $H^1(G, \mathcal{S}) \neq 0$ when restricted to any classical sub-sheaf.
\vspace{-2mm}\end{definition}
Contextual correlations provide agents with shared randomness that is fundamentally richer than any classical common knowledge. This additional structure can resolve semantic ambiguities as proved next, that would otherwise require explicit communication.
\vspace{-2mm}\subsection{ Quantifying Contextual Advantage}
To make the contextuality advantage precise, we relate the contextual fraction to semantic rate reduction. 
For a quantum semantic sheaf $\mathcal{S}$, let $e$ denote the \emph{empirical model}, which is the observed joint probability distribution over agents' semantic measurements. The \emph{contextual fraction} is:
\vspace{-2mm}\begin{equation}
\mathrm{CF}(e) = \min_{e' \in \mathrm{NC}} \|e - e'\|_1
\vspace{-1mm}\end{equation}
where $\mathrm{NC}$ is the set of non-contextual models and $\|\cdot\|_1$ is the total variation distance.
\vspace{-2mm}\begin{theorem}
Let $\mathcal{S}_C$ be a classical semantic sheaf and $\mathcal{S}_Q$ be a quantum semantic sheaf over the same graph $G$, where $\mathcal{S}_Q$ admits contextual correlations. Then:
\vspace{-2mm}\begin{equation}
\dim H^1(G, \mathcal{S}_Q) < \dim H^1(G, \mathcal{S}_C)
\vspace{-2mm}\end{equation}
whenever the contextual fraction is nonzero.
\end{theorem} 
\vspace{-1mm}\begin{IEEEproof}
For classical sheaves, the gluing axiom implies $\ker(\delta^1) = \mathrm{im}(\delta^0)$, hence $H^1 = 0$ when correlations are non-contextual. For quantum sheaves with contextual correlations, let $\{c_i\}$ be independent contextual witnesses. Each $c_i$ eliminates one generator of $H^1$:
$
\dim H^1(S_Q) = \dim H^1(S_C) - |\{c_i\}|
$.
The contextual fraction $\mathrm{CF}(e)\! >\! 0$ guarantees $|\{c_i\}| \geq 1$. 
\end{IEEEproof}
\vspace{-2mm}\begin{proposition}
 The semantic rate reduction due to contextuality satisfies:
$R_{\text{classical}} - R_{\text{quantum}} \geq CF(e) · \log_2\, \text{dim}\, H^1_{\text{classical}}$.
\end{proposition}
\vspace{-1mm}\begin{IEEEproof}
From Theorem~3, contextuality reduces $\dim H^1$. The contextual fraction $\mathrm{CF}(e)$ quantifies what fraction of classical correlations are replaced by contextual ones. Since each generator of $H^1$ eliminated saves $\log_2 \dim H^1_{\mathrm{classical}}$ bits of communication (by Theorem~1), the bound follows.
\end{IEEEproof} 
\vspace{-2mm}\subsection{ Quantum Discord and Semantic Integration}
The previous subsection showed contextuality \emph{reduces} communication requirements. We now show that discord \emph{quantifies} the holistic meaning that exists only in the joint system, i.e., semantic content that cannot be decomposed into what each agent knows independently. This links our framework to integrated information theory \cite{Tononi2016,ChristoJSAITArxiv2023}.
For a bipartite quantum state $\rho_{AB}$, the \emph{quantum discord} \cite{Ollivier2001} is:
\vspace{-1mm}\begin{equation}
\delta(A:B) = I(A:B) - J(A:B)
\vspace{-1mm}\end{equation}
where $I(A:B) = S(\rho_A) + S(\rho_B) - S(\rho_{AB})$ is the quantum mutual information, and $J(A:B) = \max_{M_B}[S(\rho_A) - S(\rho_A|M_B)]$ is the classical correlation, maximized over measurements $M_B$ on system $B$.
Discord captures quantum correlations beyond entanglement. It is nonzero for some separable states and quantifies the disturbance caused by measurement.
\vspace{-2mm}\begin{definition}
For a semantic state $\rho_{AB}$ over a bipartite semantic sheaf for any two nodes in $G$, the \emph{integrated semantic information} is:
\vspace{-2mm}\begin{equation}
\Phi_{\mathrm{sem}}(A:B) = I(A:B) - \max_{\mathrm{partitions}} \sum_i I(A_i:B_i)
\vspace{-2mm}\end{equation}
where the maximization is over all partitions of $A$ and $B$ into independent subsystems.
\vspace{-2mm}\end{definition}
Integrated semantic information measures the semantic content that cannot be reduced to independent part, which is the irreducibly holistic meaning that emerges from the interaction of subsystems.

\vspace{-2mm}\begin{theorem}
    For bipartite quantum semantic states:
$\delta(A:B) = \phi_{\text{sem}}(A:B)$.
\end{theorem} 
\vspace{-1mm}\begin{IEEEproof}
    We show the equality by demonstrating that both quantities measure the same irreducible correlations from dual perspectives.
The quantum mutual information decomposes as:
\vspace{-1mm}\begin{equation}
I(A:B) = J(A:B) + \delta(A:B),
\end{equation}
where $J(A:B) = \max_{M_B}[S(\rho_A) - S(\rho_A|M_B)]$ is the classical correlation extractable by measuring subsystem 
B, and $\delta(A
:B)$ is the quantum discord.
The classical correlation $J(A:B)$ represents the maximum information about
A obtainable through local measurements on 
B. This corresponds precisely to $\max_{\mathrm{partitions}} \sum_i I(A_i
:B_i)$ in the definition of $\Phi_{\mathrm{sem}}
$, since the optimal measurement on
B induces the partition that maximizes extractable local information.
 Substituting into the definition of $\Phi_{\mathrm{sem}}$:
\vspace{-2mm}\begin{equation}
\begin{aligned}
    \Phi_{\mathrm{sem}}(A
:B) &= I(A:B) - \max_{\mathrm{partitions}} \sum_i I(A_i:B_i) \\ &= I(A:B) - J(A:B) = \delta(A:B).
\end{aligned}
\vspace{-1mm}\end{equation} 
Thus both quantities answer the same question from dual perspectives: discord asks ``what correlations survive any local measurement?" while integrated information asks ``what information cannot be decomposed into parts?" The equality holds because the optimal measurement basis for extracting classical correlations coincides with the minimum information partition.
\end{IEEEproof}
\vspace{-4mm}\section{Conclusion}\vspace{-1mm}
We developed an information-theoretic framework for quantum SC using sheaf cohomology. Our main results show that: (1) the minimum semantic alignment rate equals $\log_2 \dim H^1$, (2) entanglement reduces this rate by providing a rigorous mechanism for shared context, (3) contextuality further reduces communication requirements, and (4) quantum discord quantifies irreducible semantic content. Future work includes experimental validation, extension to continuous-variable systems, and application to multi-agent AI alignment.

\newpage
    
\appendices

\section{Proof of Theorem~\ref{theorem_rate}}
\label{proof_Theorem1}

We provide a rigorous proof that the minimum communication rate for perfect semantic alignment equals $\dim H^1(G, \mathcal{S})$.
\subsection{Preliminaries}
Let $G = (V, E)$ be a finite graph and $\mathcal{S}$ a semantic sheaf with stalk $\mathcal{S}_v$ at each vertex $v \in V$, where each $\mathcal{S}_v$ is a finite-dimensional vector space over $\mathbb{C}$. We define the cochain spaces:
\begin{align}
C^0(G, \mathcal{S}) &= \bigoplus_{v \in V} \mathcal{S}_v \\
C^1(G, \mathcal{S}) &= \bigoplus_{e \in E} \mathcal{S}_{t(e)}
\end{align}
where $t(e)$ denotes the target vertex of edge $e$. Using the coboundary operator $\delta^0$ in \eqref{eq_coboundary}, a configuration $\omega \in C^1$ is \emph{locally consistent} if $\omega \in \ker(\delta^1)$, satisfying the cocycle condition \footnote{\footnotesize A cocycle is an element satisfying $\delta^1 \omega = 0$, meaning the edge assignments are consistent around every closed loop in the graph.} on all cycles. The \emph{semantic alignment problem} is: given $\omega \in \ker(\delta^1)$, find $\sigma \in C^0$ such that $\delta^0(\sigma) = \omega$. By definition of cohomology, such $\sigma$ exists if and only if $[\omega] = 0$ in $H^1 = \ker(\delta^1)/\mathrm{im}(\delta^0)$.


\subsection{Achievability: $\dim H^1$ Suffices}

\begin{lemma}
There exists a semantic alignment protocol using exactly $k = \dim H^1(G, \mathcal{S})$ transmitted symbols.
\end{lemma}

\begin{proof}
Let $k = \dim H^1(G, \mathcal{S})$ and choose a basis $\{[\omega_1], \ldots, [\omega_k]\}$ for $H^1$, where each $\omega_i \in \ker(\delta^1)$ is a representative cocycle.
Given any locally consistent configuration $\omega \in \ker(\delta^1)$, its cohomology class can be uniquely expressed as:
\begin{equation}
[\omega] = \sum_{i=1}^{k} c_i [\omega_i], \quad c_i \in \mathbb{C}
\end{equation}
The encoder (e.g., a designated coordinator) transmits the coefficient vector $\mathbf{c} = (c_1, \ldots, c_k) \in \mathbb{C}^k$ to all agents that require alignment.
The receiver reconstructs $\tilde{\omega} = \sum_{i=1}^{k} c_i \omega_i$. Since $[\omega] = [\tilde{\omega}]$ under noiseless transmission, we have $\omega - \tilde{\omega} \in \mathrm{im}(\delta^0)$. Thus there exists $\sigma \in C^0$ with $\delta^0(\sigma) = \omega - \tilde{\omega}$. The receiver solves this linear system (guaranteed to have a solution by construction) to obtain the global section $\sigma$ achieving semantic alignment.
Here, the protocol transmits exactly $k = \dim H^1$ symbols.
\end{proof}

\subsection{Converse: $\dim H^1$ is Necessary}
\begin{lemma}
Any semantic alignment protocol requires at least $k = \dim H^1(G, \mathcal{S})$ transmitted symbols.
\end{lemma}
\begin{IEEEproof}
Suppose, for contradiction, that a protocol achieves perfect semantic alignment using $r < k$ transmitted symbols. The encoder is then a function $\mathrm{Enc}: \ker(\delta^1) \to \mathbb{C}^r$ mapping locally consistent configurations to $r$-symbol codewords.
For perfect alignment, configurations with different cohomology classes must map to different codewords. If $[\omega] \neq [\omega']$ in $H^1$, then $\mathrm{Enc}(\omega) \neq \mathrm{Enc}(\omega')$; otherwise the decoder cannot distinguish them.
Consider the induced map:
\begin{equation}
\overline{\mathrm{Enc}}: H^1 \to \mathbb{C}^r, \quad [\omega] \mapsto \mathrm{Enc}(\omega)
\end{equation}
This is well-defined: if $[\omega] = [\omega']$, then $\omega - \omega' = \delta^0(\sigma)$ for some $\sigma$, and both configurations lead to the same aligned state without communication. For perfect alignment, $\overline{\mathrm{Enc}}$ must be injective. But an injective linear map from a $k$-dimensional space to an $r$-dimensional space requires $r \geq k$. Since $r < k$ by assumption, $\overline{\mathrm{Enc}}$ cannot be injective—there exist $[\omega] \neq [\omega']$ with identical codewords, making alignment impossible.
\end{IEEEproof}

\section{Proof of Theorem 2}
\label{proof_Theorem2}
Using the cochain complex from Appendix~\ref{proof_Theorem1}, we augment the semantic sheaf with entanglement. For each edge $e = (u, v)$, the entangled state admits Schmidt decomposition:
\begin{equation}
|\Phi_e\rangle = \sum_{i=1}^{r_e} \lambda_i |a_u^i\rangle \otimes |a_v^i\rangle
\end{equation}
with Schmidt rank $r_e$ and coefficients $\lambda_i > 0$.
\subsection{Part 1: Achievability (Upper Bound)}
We first show $\dim H^1_{EA} \leq \max(0, \dim H^1 - \sum_e \log_2 r_e)$. The Schmidt decomposition of $\Phi_e$ establishes a canonical correspondence between $r_e$ basis vectors at agents $u$ and $v$. This shared reference frame resolves ambiguity in how local semantic states are compared.
By the superdense coding protocol \cite{Bennett2002EntanglementAssisted}, $r_e$ shared dimensions enable transmission of $\log_2 r_e$ classical bits without additional channel use. These bits can specify which of $r_e$ obstruction classes applies on edge $e$.
Let $\{\omega_1, \ldots, \omega_k\}$ be generators of $H^1$ with $k = \dim H^1$. Decompose each cocycle by edge support: $\omega = \sum_e \omega_e$. The entanglement $\Phi_e$ trivializes obstructions within the $r_e$-dimensional shared subspace, converting cocycles to coboundaries in $\mathcal{S}_{EA}$. Each edge contributes resolution capacity $\log_2 r_e$. Total resolution: $\sum_{e \in E} \log_2 r_e$. The remaining obstruction satisfies:
\begin{equation}
\dim H^1_{EA} \leq \dim H^1 - \sum_{e \in E} \log_2 r_e
\end{equation}
Since dimension is non-negative, $\dim H^1_{EA} \leq \max(0, \dim H^1 - \sum_e \log_2 r_e)$.
\subsection{Part 2: Converse (Lower Bound)}
We show $\dim H^1_{EA} \geq \max(0, \dim H^1 - \sum_e \log_2 r_e)$. Entanglement on edge $e$ only affects obstructions with support on $e$. A cocycle $\omega$ with global support across multiple edges cannot be fully resolved by entanglement on any single edge.
Resolving a $k$-dimensional obstruction space requires $\log_2 k$ bits of information. The entanglement across all edges provides at most $\sum_e \log_2 r_e$ bits of shared reference. By a counting argument, at least $\dim H^1 - \sum_e \log_2 r_e$ dimensions of obstruction remain.
For any cocycle $\omega \in Z^1$ orthogonal to the subspaces spanned by entanglement-induced identifications, we have $[\omega] \neq 0$ in $H^1_{EA}$. The space of such cocycles has dimension at least $\max(0, \dim H^1 - \sum_e \log_2 r_e)$.

  \subsection{Part 3: Classical Separation}

For separable states $\rho_e = \sum_i p_i \rho_u^i \otimes \rho_v^i$, the Schmidt rank is $r_e = 1$. Thus:
\begin{equation}
\sum_{e \in E} \log_2 r_e = \sum_{e \in E} \log_2 1 = 0
\end{equation}
Therefore $\dim H^1_{\mathrm{classical}} = \dim H^1$: classical correlations provide zero obstruction reduction.

This completes the proof.

\bibliographystyle{IEEEtran}
\bibliography{semantics_ref}

\end{document}